\documentclass[11pt, twocolumn]{article}
\usepackage[top=1in, bottom=1in, left=1in, right=1in]{geometry}
\usepackage{amsmath}
\usepackage{amsfonts}
\usepackage{amssymb}
\usepackage{amsthm}
\usepackage[usenames, dvipsnames]{xcolor}
\usepackage{colortbl}
\usepackage{tikz}
\usetikzlibrary{automata,topaths,arrows,shapes,positioning,calc,decorations.pathreplacing}
\usepackage{graphicx}
\usepackage{ulem}
\usepackage{caption}
\usepackage{pgfplots}
\usepackage{float}
\usepackage{mathrsfs}
\usepackage{wrapfig}
\usepackage{hyperref}
\usepackage{mathdots}
\usepackage{arydshln}
\usepackage{subcaption}
\usepackage[font=footnotesize,labelfont=bf]{caption}
\usepackage{lscape}
\usepackage{enumitem}
\usepackage{placeins}
\usepackage[affil-it]{authblk}
\usepackage{abstract}
\usepackage[multiple]{footmisc}
\usepackage{titlesec}
\usepackage{bibunits}

\theoremstyle{plain}  
\newtheorem{thm}{Theorem}[section]
\newtheorem{lem}[thm]{Lemma}

\newtheorem{cor}[thm]{Corollary}

\newtheorem{conjecture}[thm]{Conjecture}
\theoremstyle{definition}  
\newtheorem{defn}[thm]{Definition}

\newtheorem{prob}[thm]{Problem}

\theoremstyle{remark}  

\newtheorem{rem}[thm]{Remark}

\newcommand{\N}{{\mathbb{N}}}

\definecolor{gray85}{gray}{0.85} 
\definecolor{gray8}{gray}{0.8} 
\definecolor{gray7}{gray}{0.7} 
\definecolor{gray6}{gray}{0.6} 
\definecolor{gray5}{gray}{0.5} 
\definecolor{gray4}{gray}{0.4} 
\definecolor{gray35}{gray}{0.35} 

\newcommand\blfootnote[1]{%
  \begingroup
  \renewcommand\thefootnote{}\footnote{#1}%
  \addtocounter{footnote}{-1}%
  \endgroup
}
\newcommand{\abs}[1]{\left| #1 \right|} 
\usepackage{color}
\definecolor{forestgreen}{rgb}{0.13, 0.55, 0.13}

\definecolor{darkgray}{rgb}{.55,.55,.55}
\definecolor{lightgray}{rgb}{.8,.8,.8}
\definecolor{verylightgray}{rgb}{1,1,1}

\newtheorem{question}[]{Question}
\newtheorem{mainresult}{Main Result}

\title{Genetic Networks Encode Secrets of Their Past}
\author{Peter Crawford-Kahrl\textsuperscript{1,\textdagger}, Robert R. Nerem\textsuperscript{2,\textdagger}, Bree Cummins\textsuperscript{1}, and Tomas Gedeon\textsuperscript{1}}
\date{}
\begin{document}

\twocolumn[
  \begin{@twocolumnfalse}
    \maketitle
    \renewcommand{\abstractname}{Significance Statement}
    \begin{abstract}
The study of gene regulatory networks has expanded in recent years as an abundance of experimentally derived networks have become publicly available. The sequence of evolutionary steps that produced these networks are usually unknown.  As a result, it is challenging to differentiate  features that arose through gene duplication and gene interaction removal from features introduced through other mechanisms. We develop tools to  distinguish these network features and in doing so, give methods for studying ancestral networks through the analysis of present-day networks. 
    \end{abstract}
    
    \renewcommand{\abstractname}{Abstract}
    \begin{abstract}
    Research shows that gene duplication followed by either repurposing or removal of duplicated genes is an important contributor to evolution of gene and protein interaction networks. We aim to identify which characteristics of a  network can arise through this process, and which must have been produced in a different way.  To model the network evolution, we postulate vertex duplication and edge deletion as evolutionary operations on graphs. Using the novel concept of an \emph{ancestrally distinguished subgraph}, we show how features of present-day networks require certain features of their ancestors. In particular,  ancestrally distinguished subgraphs cannot be introduced by vertex duplication.  Additionally, if vertex duplication and edge deletion  are   the only evolutionary mechanisms, then a graph's ancestrally distinguished subgraphs must be contained in  all of the graph's ancestors.
    We analyze two experimentally derived genetic networks and show that our results accurately predict lack of large ancestrally distinguished subgraphs, despite this feature being statistically improbable in  associated  random networks. This observation is consistent with the hypothesis that these networks evolved primarily via vertex duplication.  The tools we provide open the door for analysing ancestral networks using  current networks. Our results apply to edge-labeled (e.g. signed) graphs which are either undirected or directed.
      \newline
      \newline
    \end{abstract}

  \end{@twocolumnfalse}
]

\section{Introduction}

\blfootnote{\textsuperscript{1}Department of Mathematical Sciences,\\
Montana State University,
Bozeman, Montana, USA\\
\indent\indent\textsuperscript{2}Institute for Quantum Science and Technology,\\ University of Calgary, Alberta T2N 1N4, Canada \\
\indent\indent\textsuperscript{\textdagger}These authors contributed equally to this work}
Gene duplication is one of the most important mechanisms governing genetic network growth and evolution~\cite{  li1997molecular, ohno2013evolution, patthy2009protein}. Another important process is the elimination of interactions between existing genes, and even entire genes themselves. These two mechanisms are often linked, whereby a duplication event is followed by the removal of some of the interactions between the new gene and existing genes in the network~\cite{conant2003asymmetric, dokholyan2002expanding, Janwa2019, taylor2004duplication, vazquez2003modeling, wolfe}. De novo establishment of new interactions or addition of new genes into the network by horizontal gene transfer is also possible, but significantly less likely~\cite{Wagner03}. 

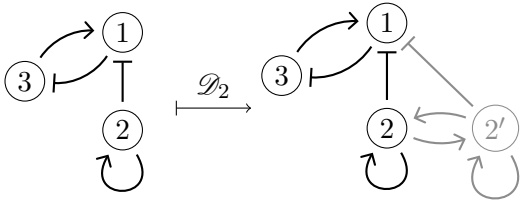
\begin{figure}[t]
\begin{center}
\begin{tikzpicture}
 
    \node[] (1) at (0,0) 
    {
    \begin{tikzpicture}[main node/.style={circle,draw,font=\sffamily\normalsize\bfseries}, inner sep=2pt, minimum size=15pt, scale=1.3]
		
		\node[main node] (c) at (0,0) {$2$};
		\node[main node] (a) at (0,1) { $1$};
		\node[main node] (b) at (-1,.5) { $3$};
	
		\path[->,thick,shorten >= 3pt,shorten <= 3pt,>=angle 90]
		(a) edge[bend left,-|] (b)
		(b) edge[bend left] (a)
		(c) edge[-|] (a)
		(c) edge[in=-120,out=-60,loop, looseness=8] (c)
		;
		\end{tikzpicture}	
		};
    \node[] (2) [right=of 1] 
    {
    \begin{tikzpicture}[main node/.style={circle,draw,font=\sffamily\normalsize\bfseries}, inner sep=2pt, minimum size=15pt, scale=1.4]
		
		\node[main node] (c) at (0,0) { $2$};
		\node[main node] (a) at (0,1) { $1$};
		\node[main node] (b) at (-1,.5) { $3$};
		\node[main node, draw=darkgray] (c') at (1,0) {\color{darkgray}  $2'$};
		
		\path[->,thick,shorten >= 3pt,shorten <= 3pt,>=angle 90]
		(a) edge[bend left,-|] (b)
		(b) edge[bend left] (a)
		(c) edge[-|] (a)
		(c) edge[in=-120,out=-60,loop, looseness=8] (c)
		(c) edge[bend right=20,darkgray] (c')
		(c') edge[bend right=20,darkgray] (c)
		(c') edge[in=-120, out=-60, loop, looseness=8,darkgray] (c')
		(c') edge[-|,darkgray] (a)
		;
		\end{tikzpicture}
		};
        
		\path[|->] (1) edge [] node[above]  {$\mathscr{D}_2$} (2);
\end{tikzpicture}
\end{center}
\caption{An illustration of vertex duplication. The left graph is $G$, and the right graph is $G'=\mathscr{D}_2(G)$. Here, vertex 2 is duplicated, resulting in the addition of vertex $2'$ and new edges, all of which are shown in grey. Vertex $2'$ inherits all of the connections of vertex 2. Since $2$ possesses a self-loop, we add connections between $2$ and $2'$.}
\label{fig:nodeDuplication}
\end{figure}

A common description of protein-protein interaction networks and genetic regulatory networks is that of a graph. Several papers study how gene duplication, edge removal  and vertex removal affect the global structure of the interaction network from a graph theoretic perspective~\cite{vazquez03,dorog01,sole02,Wagner01,Wagner03}. They study the effects that the probability of duplication and removal have on various network characteristics, such as the degree distribution of the network. 
These papers conclude that by selecting proper probability rates of vertex doubling, deletion of newly created edges after vertex doubling, and addition of new edges, one can recover the degree distribution observed in inferred genetic networks in the large graph limit. This seems to be consistent with the data from \textit{Saccharomyces cerevisiae}~\cite{Wagner01,Wagner03} but since regulatory networks are finite, the distributions of genetic networks are by necessity only approximations to the theoretical power distributions.

Other investigations are concerned with general statistical  descriptors of large networks. These descriptors include the distribution of path lengths, number of cyclic paths, and other graph characteristics ~\cite{albert02,Barabasi99,Jeong01,watts99}. These methods are generally applicable to any type of network (social interactions, online connections, etc) and are often used to compare networks across different scientific domains.

We take a novel approach to analyzing biological network evolution. We pose the following question: 

\begin{question}\label{question}
Given a current network, with no knowledge of its  evolutionary path, can one recover structural traces of its ancestral network?
\end{question}


To answer this question we formulate a general model of graph evolution, with two operations: the duplication of a vertex  and removal of existing vertices or edges. The effect of vertex duplication, shown in Figure ~\ref{fig:nodeDuplication}, is defined by a vertex and its duplicate sharing  the same adjacencies. This model does not put any constraints on which vertices or edges may be removed, the order of evolutionary operations, nor limits the number of operations of either type. Previous investigations of the evolution of networks under vertex duplication study special cases of our model~\cite{conant2003asymmetric, dokholyan2002expanding,taylor2004duplication, vazquez2003modeling}. 

Suppose that a particular sequence of evolutionary operations transforms a graph $G$ into a graph $G'$. We seek to discover which characteristics and features of the ancestor $G$ may be recovered from knowledge of $G'$. Although this work is motivated by biological applications, the results in our paper apply to any edge-labeled directed or undirected graph.

 Our results are in two related directions. First, we introduce the concept of a ancestrally distinguished subgraph and show that $G$ must contain all (ancestrally) distinguished subgraphs of $G'$. This implies that vertex duplication and edge deletion can not introduce distinguished subgraphs. Next, we define the distinguishability of graph as the size of of its largest distinguished subgraph. Our theoretical analysis suggests that small distinguishability is a signature of networks that evolve primarily via vertex duplication. We confirm this result by showing that the distinguishabilities of two published biological networks and artificial networks evolved by simulated vertex duplication both exhibit  distinguishability that is smaller than their expected distinguishability under random edge relabeling.

\section{Main Results}
\label{sec:Results}
\subsection{Ancestral Networks Contain Distinguished Subgraphs}

We begin by introducing a new graph property that we call \emph{ancestral distinguishability} (Definition \ref{def:distinguishability}) shortened to distinguishability hereafter. We say two vertices are distinguishable if there exists a mutual neighbor for which the edges connecting the vertices to this neighbor have different edge labels. In a directed graph, a mutual neighbor is either a predecessor of both vertices or a successor of both vertices. Since, by definition of duplication, a vertex and its duplicate must be connected to each of their neighbors by edges with the same label  (Figure \ref{fig:nodeDuplication}, Definition \ref{defn:nodeDouble}), we show that a vertex and its duplicate can never be  distinguishable. Additionally, deletion of edges can not create distinguishability between two vertices. 

We combine these results to prove that vertex duplication and edge deletion cannot  create  new subgraphs for which every pair of vertices is distinguishable. This observation yields our first main result that any such \emph{distinguished subgraph} in the current network $G'$, must have also occurred in the ancestral network $G$ (Corollary \ref{cor:primordialSubgraph}). In fact this result is a corollary of a stronger theorem regarding the existence of a certain graph homomorphism from $G'$ to $G$ (Theorem~\ref{thm:lookBackMap}).

\begin{mainresult} \label{mr:1}
If $G'$ is a network formed from $G$ by vertex duplication and edge deletion, then all distinguished subgraphs of $G'$ are isomorphic to distinguished subgraphs of $G$.  In other words, no distinguished subgraph in $G$ could have been introduced by vertex duplication and edge deletion.
\end{mainresult}


We develop Main Result \ref{mr:1} in the setting for which vertex duplication and edge deletion are the only evolutionary mechanisms. However, if there are evolutionary mechanisms other than vertex duplication and edge deletion, the the second formulation of Main Result 1 offers an important insight. If a sequence of arbitrary evolutionary steps (vertex duplication, edge deletion, or some other mechanism) takes a network $G$ to a network $G'$ containing a  distinguished subgraph $H$, then either $H$ is isomorphic to a subgraph of $G$ or
at least one step in the evolutionary sequence was not vertex duplication or edge deletion.

\subsection{A Robust Signature of Duplication}

\begin{figure}
    \centering
    \includegraphics[width=\linewidth]{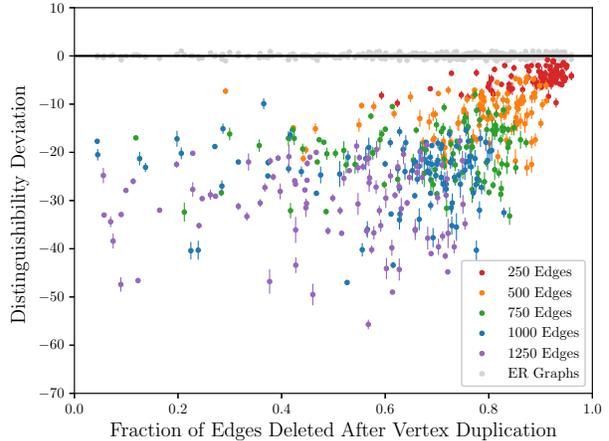}
    \caption{ Colored points represent 500 directed graphs generated  from random 25-vertex seed graphs by repeated random vertex duplication and subsequent edge deletion until a predetermined number of edges is achieved. Color indicates final number of edges after deletion. Each of the 500 grey points  represents a randomly generated  ER-graph with number of vertices, positive edges, and negative edges equal to that of a corresponding evolved graph. The corresponding figure for undirected graphs is Figure \ref{fig:und fig1} in the SI. }
    \label{fig:sfig1}
\end{figure}

We next aim to determine if the effects of evolution by vertex duplication and edge deletion can be identified in biological networks. We consider the \emph{distinguishability} of a graph, which is the number of vertices in its largest distinguished subgraph. Since vertex duplication and edge deletion cannot create distinguishability, the distinguishability of a graph cannot increase under this model of evolution (Corollary \ref{cor:disNumberConserved}). Since observations indicate that evolution is dominated by duplication and removal, we predict that genetic networks exhibit low distinguishability. 

To quantify the degree to which the distinguishability of a graph $G$ is low, we compute the \emph{distinguishability deviation} of $G$: the difference between the distinguishability of $G$ and the expected distinguishability of $G$ under random edge relabeling (Equation \ref{eq:rel dist}).
 Since low distinguishability is a signature of vertex duplication, we expect random relabeling to remove this signature and therefore increase distinguishability. In other words, we expect networks evolved by vertex duplication and edge deletion to have negative distinguishability deviation. 

We calculate the distinguishability deviation of networks constructed by simulated evolution via vertex duplication and edge deletion.
 These networks are formed in two stages from  25-vertex Erd{\"o}s-R{\'e}nyi graphs (ER-graphs \cite{ER}) with two edge labels denoting positive and negative interaction.  First, vertex duplication is applied 225 times, each time to a random vertex. Next, edges are randomly deleted until some target final number of edges is reached. The deletions simulate both evolutionary steps and the effect of incomplete data in experimentally derived networks.  We note that the operation of vertex duplication and edge removal commute in a sense that any graph that can be built by an arbitrary order of these operations can be also built by performing the duplications first and then performing an appropriate number of deletions. Therefore our construction is general. 

As shown in Figure \ref{fig:sfig1}, these simulations indicate that networks evolved by vertex duplication have negative distinguishability deviation. For each graph represented by a colored point in Figure \ref{fig:sfig1}, we construct an ER-graph with the same number of vertices, positive edges, and negative edges. These graphs are represented by grey points and show that ER-graphs exhibit near-zero distinguishability deviation. This negativity is robust against edge deletion; even graphs that had 80\% of their edges deleted after vertex duplication exhibited statistically significant negative distinguishability deviation.  

Having established evidence that graphs evolved by vertex duplication exhibit negative distinguishability deviation, we  evaluate if this property is observable in biological networks. 
We consider two networks. The first is a \textit{D. melanogaster} protein-protein interaction network developed by~\cite{vin14}, represented by an edge-labeled undirected graph. Second, we investigate the directed human blood cell regulatory network recorded in \cite{Collombet2017}. Both networks have label set $L=\{-1,+1\}$,  signifying negative and positive  regulation, respectively. 

The  distinguishability deviations of these networks confirm our predictions. Respectively, the distinguishabilities of the \textit{D. melanogaster} and blood cell networks are 7 and 4 and their expected distinguishabilities approximated by 100 random edge sign relabeling  are  $ 31.2 \pm .7 \; \mbox{ and } \;5.6 \pm .6 $. Thus, these networks have distinguishability deviations of
\begin{equation}\label{eq:RelativeDistRealNet}
 -24.2 \pm .7 \; \mbox{ and } \;-1.6 \pm .6 
\end{equation}
with  statistical significance of $34.6$ and $2.3$ standard deviations, respectively. 
%
These results are consistent with the hypothesis that biological networks inferred from experimental data are subject to long sequences of vertex duplication and edge removal without the evolutionary operation of novel vertex or edge addition.


 
The joint evidence of negative distinguishability deviations in both simulated and observed data leads to the following result.

\begin{mainresult}
Negative distinguishability deviation is a likely signature of evolution via vertex duplication and edge deletion.
\end{mainresult}

While we do not offer a rigorous mathematical proof, in Subsection \ref{sec:resultsComputation} we give evidence for a conjecture (Conjecture~\ref{theMainConjecture}) which, if true, would prove that vertex duplication   always  decreases distinguishability deviation. 
SI Section~\ref{app:numerical} gives a detailed description of the simulated evolution scheme we used in Figure \ref{fig:sfig1}. For completeness, we show in this section that negative distinguishability deviation cannot be fully explained by the single vertex characteristics (i.e. signed degree sequence) or small world properties of the networks.

\section{Discussion}

We introduce the concept of distinguished subgraphs, in which every vertex has differentiating regulatory interactions from every other vertex in the subgraph. We show that distinguished subgraphs cannot be created by vertex duplication and edge deletion. Remarkably, this implies that any of a network's distinguished subgraphs must appear in all of its ancestors under a model of network evolution that allows duplication and removal, but does not allow for the addition of new vertices or edges. Furthermore, this result shows that distinguished subgraphs cannot be introduced by vertex duplication and edge deletion. 

In biological networks the addition of regulatory interactions between existing genes (neofunctionalization \cite{Force1999}), or the addition of entirely new genes via horizontal gene transfer~\cite{Wagner03} are possible, but are considered less likely than gene duplication or loss of function of a regulatory interaction \cite{Bergthorsson2007}. With this in mind, we consider a model of network evolution in which long sequences of vertex duplication and edge removal are interspersed by infrequent additions of new edges or vertices. Under this model, Main Result~\ref{mr:1} (Corollary~\ref{cor:primordialSubgraph}) applies to any sequence of consecutive vertex duplications and edge removals.

 
We investigate whether the predicted features of vertex duplication can be found in biological networks inferred from experimental observations. Using the metric of distinguishability deviation we show that two inferred biological networks and a population of simulated networks evolved by vertex duplication exhibit negative distinguishability deviation that is statistically improbable in associated random networks. We propose that negative distinguishability deviation is a marker of evolution by vertex duplication and edge removal. 

One potential application of this result is a method of checking the suitability of random graph models. Often, random statistical models are developed to generate graphs that match properties of social networks \cite{newman2002random}, properties of biological networks \cite{saul2007exploring}, or general graph theoretic properties \cite{fosdick2018configuring}. For example, the discovery of small-world phenomena~\cite{Milgram1967,watts99} lead to the development of the Watts-Strogatz model~\cite{Watts1998}.  Our results imply that an accurate random graph model for signed biological networks, or more generally edge-labeled networks that primarily evolved via vertex duplication, should generate networks with negative distinguishability deviation. Additionally, distinguishability deviation could inform the development of new models that more closely agree with experimentally derived networks.

As an illustration of the utility of Main Result \ref{mr:1}, we consider the following example. Certain network motifs, i.e.\ 3-4 vertex subgraphs, have been shown to appear at statistically higher rates in inferred biological networks~\cite{milo2002network}. Motifs seem to be a byproduct of convergent evolution, being repeatedly selected for based on their underlying biological function, and appearing in organisms and systems across various biological applications~\cite{alon2007network}. 

Vertex duplication and edge removal can easily create new motifs. For example, consider the feed-forward loop, any three vertex subgraph isomorphic to a directed graph with edge set $\{ (i,j), (j,k), (i,k) \}$ (see~\cite{shen2002network}). In Figure~\ref{fig:nodeDuplication}, no feed-forward loops can be found in $G$, but there are two in $G'$, both of which contain the vertices $1$, $2$, and $2'$. In contrast, the introduction of motifs that are also distinguished subgraphs by vertex duplication and edge deletion is forbidden by Main Result~\ref{mr:1}. Indeed,  the feed-forward loops created in Figure~\ref{fig:nodeDuplication}  are not distinguished subgraphs.
This ability to identify which motifs could not have arisen from vertex duplication and edge deletion could provide new insight into the origin of specific motifs and, potentially, their biological importance.
Similarly, identifying genes in subgraphs that cannot arise from vertex duplication and edge deletion could be useful for finding genes that were introduced by mechanisms outside of these operations, such as horizontal gene transfer.

Finally, our mathematical results are general enough to survey network models beyond genetics to discern if vertex duplication may have played a role in their evolution.
For example, current ecological networks reflect past speciation events, where a new species  initially shares the ecological interactions of their predecessors. This can be viewed as vertex duplication and therefore ecological networks may exhibit significant negative distinguishability deviation. Evaluating the distinguishability deviation of ecological networks could indicate if the duplication process has been a significant factor in their evolution. More broadly, the study of the evolutionary processes that produce networks has been used to understand why networks from distinct domains, be they social, biological, genetic, internet connections, etc, have properties unique to their domain (e.g. exponents of power law distributions  \cite{Graham2003}).  
 Distinguishability deviation is yet another tool to understand the effect evolutionary processes have on networks.



\section{Methods}

We proceed with preliminary definitions to familiarize the reader with the language and notation used in this paper.

\subsection{Definitions}

Throughout this paper we fix an \emph{edge label set} $L$. We assume that $\abs{L}\geq 2$, otherwise the results are trivial. For example, to consider signed regulatory networks with both activating and inhibiting interactions one could take $L =\{+1,-1\}$. We use this choice in examples, along with the notation $\dashv$  and $\to$ to represent directed edges with labels $-1$ and $+1$ respectively. 

\begin{defn}
A \emph{graph} is the 3-tuple $G := (V,E,\ell)$ where $V$ is a set of vertices, $E \subseteq \{(i,j) : i,j\in V\}$ is a set of directed edges, and $\ell: E \to L$ is a map labeling edges with elements of $L$. 
%
\end{defn}

\noindent Our results apply to both directed graphs and undirected graphs. To facilitate this, we use graph to mean either an undirected or directed graph, and view undirected graphs as a special case of directed graphs, as seen in the following definition.

\begin{defn} \label{defn:undirectedGraph}
A graph $G = (V,E,\ell)$ is \emph{undirected} if $(i,j) \in E$ and $\ell(i,j)= a$ if and only if $(j,i) \in E$ and $\ell(j,i)=a$. For an unlabeled graph, $\ell = \emptyset$.
\end{defn}

\begin{defn}
A \emph{subgraph} of a graph $G = (V,E,\ell)$ is a graph $H = (V',E',\ell|_{E'})$ such that $V' \subseteq V$ and $E' \subseteq E \cap V'\times V'$. If $H$ is undirected, we require that $G$ is also undirected, i.e.\ $E'$ satisfies $(i,j)\in E$ if and only if $(j,i)\in E$.
\end{defn}

\begin{defn}
Let $(V,E,\ell)$ be a graph. We say $j \in V$ is a \emph{neighbor} of $i \in V$ if either $(j,i) \in E$ or $(i,j) \in V$.
\end{defn} 

\begin{defn}
Let $G'=(V',E',\ell')$ and $G= (V,E,\ell)$ be two graphs. A map $\Phi\colon V' \to V$ is a \emph{graph homomorphism (from $G'$ to $G$)} if $\forall i,j \in V'$, if $(i,j) \in E'$, then  $(\Phi(i),\Phi(j)) \in E$ and $\ell'(i,j) = \ell(\Phi(i),\Phi(j))$. In other words, a graph homomorphism is a map on vertices that respects edges and edge labels.
\end{defn}

The following definition specifies an operation on a graph which duplicates a vertex $d$, producing a new graph that is identical in all respects except for the addition of one new vertex, $d'$, that copies the edge connections of $d$. This definition captures the behavior of gene duplication in genetic networks.

\begin{defn}\label{defn:nodeDouble}
Given a graph $G= (V,E,\ell)$ and a vertex $d\in V$, we define the \emph{vertex duplication of $d$} as the graph operation which constructs a new graph, denoted ${\mathscr{D}_d(G) := G' = (V',E',\ell')}$, where $V' := V \cup \{ d'\}$, and $(i,j) \in E'$ with $\ell'(i,j) = a$ if and only if either
\begin{enumerate}
    \item $(i,j) \in E$ with $\ell(i,j) = a$,
    \item $j = d'$ and $(i,d) \in E$ with $\ell(i,d) = a$,
    \item $i = d'$ and $(d,j) \in E$ with $\ell(d,j) = a$,
    \item or $j=i=d'$ and $(d,d) \in E$ with $\ell(d,d) = a$. 
\end{enumerate}
\end{defn}


An example of vertex duplication is shown in Figure~\ref{fig:nodeDuplication}.

\subsection{Distinguishability}
We now introduce an important invariant property under vertex duplication and edge removal. 
\begin{defn} \label{def:distinguishability}
 Let $G = (V,E,\ell)$ be a graph. Two vertices $i,j \in V$ are \emph{distinguishable (in $G$)} if and only if there exists a vertex $k$ that is a neighbor of both $i$ and $j$ such that either 
\begin{equation} \label{eq:dist1}
(i,k),(j,k) \in E \text{ and } \ell(i,k) \neq \ell(j,k)
\end{equation}
or
\begin{equation} \label{eq:dist2}
(k,i),(k,j) \in E \text{ and } \ell(k,i) \neq \ell(k,j).
\end{equation}
We say that $k$ is a \emph{distinguisher} of $i$ and $j$. It is worth noting that there may be multiple distinguishers of $i$ and $j$, i.e.\ distinguishers need not be unique. Furthermore, if $G$ is undirected, Equation~\eqref{eq:dist1} holds for a vertex $k$ if and only if Equation~\eqref{eq:dist2} also holds.

We say $U \subseteq V$ is a \emph{distinguishable set (in G)} if for all $i,j \in U$ with $i\neq j$, the vertices $i$ and $j$ are distinguishable. Similarly, we refer to any subgraph whose vertex set is distinguishable as a \emph{distinguished subgraph}. 
\end{defn}

\begin{rem}
As long as $\abs{L}\geq 2$, for any graph $G$, there is a graph $G'$ that contains $G$ as a distinguished subgraph. To see this,  consider a subgraph $G$. Then for each pair $i,j\in G$ add a new vertex $k$ and edges $\{(i,k),(j,k)\}$ with different labels, so that $\ell(i,k) \neq \ell(j,k)$. Then $i$ and $j$ are distinguishable and $G$ is embedded as a distinguishable subgraph in a larger graph $G'$.
\end{rem}

To illustrate the concept of distinguishable sets, consider the two graphs shown in Figure~\ref{fig:nodeDuplication}. The leftmost graph has  distinguishable sets $\{1,2\}$ and $\{2,3\}$. Here, $2$ is a distinguisher of $1$ and $2$, and $1$ is a distinguisher of $2$ and $3$. However, in the rightmost graph, $2$ and $2'$ are not distinguishable. Any mutual neighbor of $2$ and $2'$ shares exactly the same edges with matching labels.
The last insight, that the duplication of a gene $d$ produces an indistinguishable pair $d$ and $d'$, is general and leads to our main result in Theorem~\ref{thm:lookBackMap}.

\subsection{Distinguished Subgraphs}\label{sec:resultsAncestor}

Fix two graphs $G$ and $G'$. Suppose that $G$ is  an \emph{ancestor} of $G'$, that is, there exists a sequence of graphs $G_1,\dots, G_M$ with $G_m:=(V_m,E_m,\ell_m)$, such that $G=G_1$, $G'=G_M$, and for each $m \in \{1,\dots,M\}$, either $G_{m+1}$ is a subgraph of $G_m$, or $G_{m+1} = \mathscr{D}_{d_m}(G_m)$, for some $d_m \in V_m$. 

To address Question~\ref{question}, we present Theorem~\ref{thm:lookBackMap}. It states that whenever $G$ is an ancestor of $G'$, then there must exist a graph homomorphism from $G'$ to its ancestor $G$ such that the homomorphism is injective on distinguishable sets of vertices. This result allows us to conclude several corollaries that characterize the properties of the ancestor network.

The proof of the following theorem makes use of  Lemma~\ref{lem:doubleGoodMap} in Appendix~\ref{section:technicalLemmas}.

\begin{thm}\label{thm:lookBackMap}
Let $G= (V,E,\ell)$ be an ancestor of $G'=(V',E',\ell')$. Then there is a graph homomorphism $\Phi\colon V' \to V$ such that for all distinguishable sets $ U \subseteq V'$, the restriction $\Phi|_U$ is 1-to-1, and $\Phi(U)$ is a distinguishable set in $G$.
\end{thm}

\begin{proof}
Let $G_1,\dots, G_M$ be the  evolutionary path connecting ancestor $G$ with the current graph $G'$, where $G_m:=(V_m,E_m,\ell_m)$. At each step, we construct a map $\Phi_m$ from $G_{m+1}$ to $G_m$ satisfying the required conditions. The composition $\Phi := \Phi_{1} \circ \dots \circ \Phi_{M-1}$ then verifies the desired result.

We now construct $\Phi_{m}$. If $G_{m+1}$ is a subgraph of $G_m$, let $\Phi_m$ be the inclusion map $\iota\colon V_{m+1} \hookrightarrow V_m$. The inclusion map is obviously a graph homomorphism, and is injective on all of $V_{m+1}$. Let $i,j\in V_{m+1}$ be distinguishable vertices in $G_{m+1}$, and let $k$ be a distinguisher of $i$ and $j$. Since $\iota$ is a homomorphism, $\iota(k) = k \in V_m$ is a distinguisher of $\iota(i),\iota(j) \in V_m$. 

If $G_{m+1} = \mathscr{D}_{d_m}(G_m)$, let $\Phi_m\colon V_{m+1} \to V_m$ be defined as
\[
\Phi_m(i) := \begin{cases}
d_m & \text{if } i=d_m'\\
i & \text{otherwise}
\end{cases} \ .
\]
We verify by using Definition~\ref{defn:nodeDouble} that this map satisfies the required properties in Lemma~\ref{lem:doubleGoodMap}. 
\end{proof}

It is worth noting that the proof of Theorem~\ref{thm:lookBackMap} is constructive; however, the construction relies on the knowledge of the specific evolutionary path, i.e a sequence of events that form the graph sequence $G_1,\dots, G_M$. In almost all applications, this sequence is unknown or only partially understood. However the existence of the homomorphism allows us to conclude features of $G$ using knowledge  of the graph $G'$. 

\begin{cor}\label{cor:primordialSubgraph}
Let $G$ be the ancestor of $G'$. Any distinguished subgraph of $G'$ is isomorphic to a subgraph of $G$. 
\end{cor}

\begin{proof}
Consider a distinguished subgraph of $G'$ with vertex set $U\subseteq V'$. Since $U$ is distinguishable, by Theorem~\ref{thm:lookBackMap} $\Phi|_U$ is an injective graph homomorphism, so it is an isomorphism onto its image. Therefore, $\Phi|_U$ is the desired isomorphism.
\end{proof}

This result describes structures that must have been present in any ancestor graph $G$, and puts a lower bound on the size of $G$. 

\begin{defn}
The \emph{distinguishability} of a graph $G = (V,E,\ell)$ is the size of a maximum distinguishable subset $U \subseteq V$. Let $\mathtt{D}(G)$ denote the distinguishability of a graph $G$. 
\end{defn}

\begin{cor}\label{cor:disNumberConserved}
Let $G$ be the ancestor of $G'$. The distinguishability of $G$ is greater than or equal to the distinguishability of $G'$,
$$\mathtt{D}(G) \geq \mathtt{D}(G').$$
\end{cor}

\begin{proof}
Let $U\subseteq V'$ be a distinguishable set in $G'$. Then $\Phi(U)$ is distinguishable in $G$, and since $\Phi|_U$ is injective, $\abs{\Phi(U)} = \abs{U}$.
\end{proof}

Identifying distinguishable sets can be  computationally  challenging, and so we recast the problem of finding distinguishable sets in terms of a more familiar computational problem. We construct a new graph whose cliques are distinguishable sets of the original graph.

\begin{defn}
The \emph{distinguishability graph} of $G = (V,E,\ell)$ is a undirected graph $D(G) := (V,E^\ast,\emptyset)$ where $(i,j) \in E^\ast$ if and only if $i$ and $j$ are distinguishable in $G$. 
\end{defn}

Recall that a set of vertices is distinguishable if and only if each pair of vertices in that set is distinguishable.  Therefore distinguishable sets in $G$ are cliques in the distinguishability graph $D(G)$, see SI Section~\ref{sec:computationComplexity}. We also prove that the clique problem is efficiently reducible to calculating the distinguishability of a graph. Since it is easy to show computing distinguishability is in the class $\mathcal{NP}$, this reduction implies that computing the distinguishability is $\mathcal{NP}$-complete.

\subsection{Distinguishability Deviation}\label{sec:resultsComputation}
We now search for  consequences of Corollary \ref{cor:disNumberConserved}  in inferred biological networks. To do so, we seek a metric that evaluates how the distinguishability of a network compares with expected distinguishability in an appropriately selected class of  random graphs. Since vertex duplication cannot increase distinguishability, we expect genetic networks to exhibit low distinguishability when compared with  similar random  graphs. The most obvious graphs to compare against are those with the same structure as $G$, and with the same expected fraction of positive and negative edges as $G$, but in which each edge has a randomly assigned label. Before formalizing this notion in Definition \ref{def:expected dist}, we adjust our perspective on undirected graphs in order to reduce notational complexity. For the rest of this manuscript, we adopt the convention that if $E$ is an edge set for an undirected graph, then $E \subseteq \{\{i,j\} : i,j \in V\}$, i.e.\ edges of undirected graphs are unordered pairs of vertices. The notation $e \in E$ then refers to $e = (i,j)$ in a directed graph and $e=\{i,j\}$ in an undirected graph.

\begin{defn} \label{def:expected dist}
Let $G =(V,E,\ell)$ be a graph. We define the probability of each label in $G$ by counting its relative edge label abundance 
\begin{equation}\label{eq:relativeAbundance}
\mathbf{p}_G(a):= \frac{\abs{ \{e \in E : \ell(e) = a\}}}{|E|} \ . 
\end{equation}
Let $\{\ell_r\}_{r\in R}$ be the set of all possible edge label maps, $\ell_r\colon E \to L$, where $R$ is an index set. Denote $G_r := (V,E, \ell_r)$ to be the graph with the same vertices and edges as $G$ but with  edge labels determined by $\ell_r$. We define the \emph{expected distinguishability of $G$} as
\begin{equation}\label{eq:expDis}
\left\langle\mathtt{D}(G)\right\rangle := \sum_{r\in R} P(G_r) \mathtt{D}(G_r).
\end{equation}
where
\begin{equation}\label{eq:probabilityofGraph}
P(G_r) = \prod_{e \in E}  \mathbf{p}_G( \ell_r(e)).
\end{equation}
We interpret $P(G_r)$ as the probability of the graph $G_r$ conditioned on using the unlabeled structure of $G$.

In addition, we define the \emph{distinguishability deviation} of $G$ as the difference between its distinguishability and its  expected distinguishability, i.e.\
\begin{equation} \label{eq:rel dist}
    \mathtt D(G) - \langle \mathtt D(G)\rangle.
\end{equation}
\end{defn}

Expected distinguishability $\langle \mathtt D(G)\rangle$ can be approximated by  randomly relabeling $G$ with probability according to Equation~\eqref{eq:probabilityofGraph} and calculating the distinguishability of the resultant graph. Repeating the process multiple times and averaging yields an approximation of expected distinguishability. We utilize this method in our calculations of distinguishability deviation in Section \ref{sec:Results}. In particular, the  distinguishability deviations in Figure \ref{fig:sfig1} were  calculated by averaging over 10 random graphs. The distinguishability deviations of the biological  networks  in  Equation~\eqref{eq:RelativeDistRealNet} were found by averaging over 100 random graphs. 

The results of distinguishability deviation calculations in published biological networks and simulated networks lead us to the following conjecture.

\begin{conjecture}\label{theMainConjecture}
Let $\mathcal{G}_n$ be the set of all graphs  $G=(V,E,\ell)$ with $n$ vertices.
Let $\mathcal{U}_n \subseteq \mathcal{G}_n$ be the set of those graphs for which 
\begin{equation} \label{eq:conj}
 \frac{1}{\abs{V}} \sum_{d\in V} \left\langle\mathtt{D}(\mathscr D_d(G))\right\rangle - \left\langle\mathtt{D}(G)\right\rangle > 0;
\end{equation}
that is,  the set of graphs for which the  expected distinguishability increases under vertex duplication. Then the fraction of graphs with this property approaches $1$ for large graphs
\[ \lim_{n \to \infty} \frac{|\mathcal{U}_n|}{|\mathcal{G}_n|} = 1 .\]
\end{conjecture}

If Conjecture \ref{theMainConjecture} is true it would imply vertex duplication decreases distinguishability deviation on average for the majority of large graphs. This  follows from  Corollary \ref{cor:disNumberConserved} which shows duplication does not increase distinguishability.  Therefore,  if duplication increases expected distinguishability, it must decrease distinguishability deviation.  
 Part of the difficulty in proving Conjecture \ref{theMainConjecture} arises because the distribution of edge labels in $G' = \mathscr{D}_d (G)$ and $G$ may be significantly different, which causes the probabilities of edge label assignments  $\ell_r$ to  change significantly between $G$ and $G'$. 
 
 However, as evidence  in support  of the conjecture we prove a version of Conjecture~\ref{theMainConjecture} in SI Section~\ref{section:towards conj} for a modified expected distinguishability that is taken over a fixed probability of edge labels. To provide the main idea of the proof, 
fix a probability of edge labels, which is be used for both $G$ and $G' = \mathscr{D}_d (G)$. Let $\{ \ell_r\}$ and $\{ \ell'_s\}$ be the sets of all possible edge label maps of $G$ and $G'$ respectively, and denote $ G_r := (V,E, \ell_r)$ and $ G'_s := (V',E', \ell_s')$.
 For this fixed labeling probability, if we randomize the labels of $G$ then the probability of a specific labeling $ \ell_r: V \to L$ is the same as the probability of any labeling $ \ell_s:V' \to L$ such that $ \ell_s|_V =  \ell_r$. Therefore, the probability of a specific $ G_r$  is the same as the probability of any such $ G_s'$. Then, noting that $ G_r$ is a subgraph of $ G_s'$, it follows from Corollary~\ref{cor:disNumberConserved} with $ G_s'$ as an ancestor of $ G_r$ that $\mathtt{D}( G_s') \geq \mathtt{D} ( G_r)$, as required.

This shows that if the expected distinguishability is taken over a fixed labeling probability, then the expected distinguishability of a graph $G$ cannot be more than that of $G'$. In fact, we show in SI Section~\ref{section:towards conj} that under this assumption as long as $d'$ has at least one neighbor, then the modified expected distinguishability of $G'$ is strictly greater than that of $G$.

\appendix

\section{Proof of Lemma \ref{lem:doubleGoodMap}}\label{section:technicalLemmas}

\begin{lem}\label{lem:doubleGoodMap}
Let $G= (V,E,\ell)$ be a graph. Let $G' = \mathscr{D}_d (G) = (V',E',\ell')$, for some $d \in V$. Let $\phi \colon V' \to V$ be the map defined as
\[
\phi(i) := \begin{cases}
d & \text{if } i=d'\\
i & \text{otherwise}
\end{cases} \ .
\]
Then $\phi$ is a graph homomorphism such that for all distinguishable sets $U \subseteq V'$, the restriction $\phi|_U$ is 1-to-1, and $\phi(U)$ is a distinguishable set in $G$.
\end{lem}

\begin{proof}
We first show $\phi$ is a graph homomorphism. Let $i,j\in V'$. If $i,j\neq d'$, then $(\phi(i),\phi(j))=(i,j)$. Inspecting Definition~\ref{defn:nodeDouble} we see $(i,j)\in E$ if and only if $(i,j) \in E'$, and  $\ell(i,j) = \ell'(i,j)$. 

Now suppose $i=d'$ and $j\neq d'$. The case where $i\neq d'$ and $j= d'$ follows a symmetric argument. Suppose that $(d',j) \in E'$. Then $(\phi(d'),\phi(j))=(d,j)$, and from the construction of $E'$ in Definition~\ref{defn:nodeDouble} we see that $(d',j)\in E'$ if and only if $(d,j)\in E$. Finally, by definition, $\ell'(d',j) = \ell(d,j)$. When $i=j=d'$, the proof follows similarly.

To prove the properties of $\phi$ on a distinguishable set, we first show that $d$ and $d'$ are not distinguishable. 
Suppose by way of contradiction that $k$ is a distinguisher of $d$ and $d'$ in $G'$. From the definition of vertex duplication, if $(d,k)\in E'$, then $(d',k)\in E'$, and $\ell'(d,k) = \ell'(d',k)$. Similarly, $(k,d)\in E'$, then $(k,d')\in E'$, and $\ell'(k,d) = \ell'(k,d')$. Therefore, neither \eqref{eq:dist1} nor \eqref{eq:dist2} in Definition~\ref{def:distinguishability} can be satisfied, a contradiction. 
We conclude that $d$ and $d'$ are not distinguishable.

Let $U \subseteq V'$ be a distinguishable set. Then since $d$ and $d'$ are not distinguishable, $U$ can contain at most one of them. Notice that $\phi$ is 1-to-1 on $V \setminus \{d\}$, as well as on $V \setminus \{d'\}$. Consequently $\phi|_U$ is 1-to-1. 

Finally, we show that $\phi(U)$ is distinguishable. Let $i,j \in U$. Let $k$ be a distinguisher of $i$ and $j$. Then since $\phi$ is a graph homomorphism, it respects edge labels, so $\phi(k)$ is a distinguisher of $\phi(i)$ and $\phi(j)$.
\end{proof}

\section*{Acknowledgements} 
TG was partially supported by National Science Foundation grant DMS-1839299 and National Institutes of Health grant 5R01GM126555-01.
PCK and RRN were supported by the National Institutes of Health grant 5R01GM126555-01.
BC was supported by  National Science Foundation grant DMS-1839299. We acknowledge the Indigenous nations and peoples who are the traditional owners and caretakers of the land on which this work was undertaken at the University of Calgary and Montana State University. 
\bibliography{network}{}
\bibliographystyle{ieeetr}
\twocolumn

\section{SI: Toward Conjecture \ref{theMainConjecture}}
\label{section:towards conj}
We now prove a restricted version of Conjecture~\ref{theMainConjecture}. 

The difficulty in proving Conjecture~\ref{theMainConjecture} arises because the distribution of edge labels in a graph may significantly change after a vertex is duplicated. To avoid this, we present a more manageable version of expected distinguishability, where the expected value is taken over the same probability both before and after the vertex duplication. Recall Equation~\eqref{eq:probabilityofGraph}, which is repeated below
\begin{equation*}
P(G_r) = \prod_{e \in E}  \mathbf{p}_G( \ell_r(e)).
\end{equation*}

Notice that $P(G_r)$ depends implicitly on the original graph $G$, as the probabilities $\mathbf{p}_G$ are determined from $\ell$ in Equation~\eqref{eq:relativeAbundance}. To simplify, we fix a set of probabilities $\{\mathbf{p}(a)\}_{a\in L}$ with $\mathbf{p}(a)\geq 0$ for all $a \in L$, and $\sum_{a\in L} \mathbf{p}(a)=1$, and such that there are at least two labels $a,b\in L$ with $a \neq b$ such that $\mathbf{p}(a) >0 $ and $\mathbf{p}(b)>0$. Using this set, we redefine $P(G_r)$, the probability of choosing a graph $ G_r = (V, E,  \ell_r)$, as
\[
P( G_r) := \prod_{e \in E} \mathbf{p}( \ell_r(e)) \ .
\]
We are now equipped to present and prove a restricted version of Conjecture~\ref{theMainConjecture}. Under the redefined probabilities, Equation \eqref{eq:towardMainConjectureWithBreve} in the following lemma is analogous to showing all terms in the sum of Equation \eqref{eq:conj} in the manuscript are non-negative. Furthermore, as long as the graph $G$ contains at least one edge, at least one term is strictly greater than zero. Of course, as the size of the graph goes to infinity, the fraction of graphs with at least one edge approaches one.

\begin{lem}\label{lem:towardMainConjecture}
Let $G = (V,E,\ell)$. Let $d \in V$ be arbitrary. Let $G' = (V', E' ,\ell') = \mathscr D_d(G)$. Fix $\{\mathbf{p}(a)\}_{a\in L}$ as above. Let $\{ \ell_r\}_{r\in R}$ and $\{ \ell'_s\}_{s\in S}$ be the set of all possible edge label maps of $G$ and $G'$ respectively, for some $R$ and $S$ index sets.
 Denote $ G_r := (V,E, \ell_r)$ and $ G'_s := (V',E', \ell_s')$.
Then
\begin{equation}\label{eq:towardMainConjectureWithBreve}
\sum_{r\in R}  P( G_r) \mathtt{D}( G_r)
\leq 
\sum_{s\in S}  P( G'_s) \mathtt{D}( G'_s),    
\end{equation}
Furthermore, the inequality is strict if $d$ has at least one neighbor.
\end{lem}

\begin{proof}
We expand the right hand side of \eqref{eq:towardMainConjectureWithBreve} as
\begin{align*}
& \sum_{s\in S}  P( G'_s) \mathtt{D}( G'_s) \\
&=  \sum_{s\in S} \left(\prod_{e \in E'} \mathbf{p}( \ell'_s(e)) \right) \mathtt{D}( G'_s) \\
&=  \sum_{s\in S} \left(\prod_{e \in E} \mathbf{p}( \ell'_s(e))\right)\left( \prod_{e \in E' \setminus E} \mathbf{p}( \ell'_s(e))\right) \mathtt{D}( G'_s) \\
\end{align*}
We now make a key observation: for each $s\in S$, there exists a unique $r \in R$ such that $ \ell_r =  \ell'_s|_E$, so  define a map $\xi : S \to R$ via $\xi (s) = r$ if and only if $ \ell_r =  \ell'_s|_E$. 
In what follows we use the more compact notation $\xi s = r$. With this insight, note that
\[
\prod_{e \in E} \mathbf{p}( \ell'_s(e)) =  \prod_{e \in E} \mathbf{p} ( \ell_{\xi  s}(e)) =  P( G_{\xi s})
\]
We  continue to rewrite the right hand side as
\begin{align*}
&  \sum_{s\in S} \left(\prod_{e \in E} \mathbf{p}( \ell'_s(e)) \right)\left( \prod_{e \in E' \setminus E} \mathbf{p}( \ell'_s(e))\right) \mathtt{D}( G'_s) \\
&=  \sum_{s\in S}  P( G_{\xi s} )\left( \prod_{e \in E' \setminus E} \mathbf{p}( \ell'_s(e))\right) \mathtt{D}( G'_s) \\
&= \sum_{r\in R} \sum_{\substack{s\in S \\ \xi s =r} }  P( G_{\xi s} )\left( \prod_{e \in E' \setminus E} \mathbf{p}( \ell'_s(e)) \right) \mathtt{D}( G'_s) \\
&= \sum_{r\in R}  P( G_{r} ) \sum_{\substack{s\in S \\ \xi s =r} } \left( \prod_{e \in E' \setminus E} \mathbf{p}( \ell'_s(e))\right) \mathtt{D}( G'_s) \\
\end{align*}
Note that $ G_{\xi s}$ is a subgraph of  $ G'_{s}$. Then applying Corollary~\ref{cor:disNumberConserved} we have $\mathtt{D}( G'_s) \geq \mathtt{D}( G_{\xi s})$. Therefore
\begin{align*}
& \sum_{r\in R}  P( G_{r} ) \sum_{\substack{s\in S \\ \xi s =r} } \left( \prod_{e \in E' \setminus E} \mathbf{p}( \ell'_s(e))\right) \mathtt{D}( G'_s) \\
&\geq \sum_{r\in R}  P( G_{r} ) \sum_{\substack{s\in S \\ \xi s =r} } \left( \prod_{e \in E' \setminus E}  \mathbf{p}( \ell'_s(e))\right) \mathtt{D}( G_{\xi  s}) \\
&= \sum_{r\in R}  P( G_{r} ) \sum_{\substack{s\in S \\ \xi s =r} } \left( \prod_{e \in E' \setminus E}  \mathbf{p}( \ell'_s(e))\right) \mathtt{D}( G_r) \\
&= \sum_{r\in R}  P( G_{r} ) \mathtt{D}( G_r) \sum_{\substack{s\in S \\ \xi s =r} } \left( \prod_{e \in E' \setminus E}  \mathbf{p}( \ell'_s(e))\right)  \\
&= \sum_{r\in R}  P( G_{r} ) \mathtt{D}( G_r)
\end{align*}
which is the desired result. The last inequality holds because for any fixed $r$
\begin{align*}
\sum_{\substack{s\in S \\ \xi s =r} } \left( \prod_{e \in E' \setminus E}  \mathbf{p}( \ell'_s(e))\right) = 1
\end{align*}
since the sum can be thought of as over all possible relabelings of $E' \setminus E$, of which the total probability is $1$.

We now  show that the inequality is strict if $d$ has at least one neighbor. To do so we construct a specific pair $ \ell'_q$ and $\ell_{\xi q}$ so that $\mathtt{D}( G'_q) = \mathtt{D}( G_{\xi q}) + 1$. Let $\mathbf{p}(a)$ and $\mathbf{p}(b)$ be two non-zero elements.  Let $k\in V$ be a neighbor of $d$. Let $ \ell'_q$ be the constant map on $a \in L$, except for the edge between $d'$ and $k$, where it takes the value $b\in L$. Then $\mathtt{D}( G'_q) =2$, as $d$ and $d'$ are distinguishable and no other vertex can be distinguishable from both $d$ and $d'$. Since $\ell_{\xi q}$ is the constant map on $a$, $\mathtt{D}( G_{\xi q})=1$, completing the proof. 
\end{proof}

\section{SI: Computational Complexity}\label{sec:computationComplexity}
 For notational convenience we return to Definition \ref{defn:undirectedGraph} for our definition of undirected graph. 

\begin{defn}
Let $G= (V,E,\ell)$ be a undirected graph. A subset $U \subseteq V$ is a \emph{clique} if for all $i,j \in U$ the vertices $i$ and $j$ are neighbors. We refer to a clique $U$ of size $|U| = m$ as a $m$-clique. 
\end{defn}

\begin{prob}[Distinguishability]
Given a graph $G = (V,E,\ell)$, a number $m \in \N$, and a label set $L$ with $|L| \leq |E|$, decide if $G$ contains a  distinguishable set of size $m$. Let \textsc{Distinguishability} be the associated decision problem represented as a language.  
\end{prob}

Although we do not specify a particular way of encoding graphs as strings in the language \textsc{Distinguishability}, any encoding that is polynomial in $|V|$ is sufficient. For discussion on decision problems as languages and graph encodings see \cite{Barak2016}.




\begin{thm}\label{lem:findClique}
A graph $G$ has a distinguishable set of size $m$ if and only if $D(G)$ has an $m$-clique.
\end{thm}
\begin{proof}
$(\Rightarrow)$ Let $U \subset V$ be a distinguishable set in $G$ of size $m$. Then all pairs $i,j \in U$ are distinguishable which implies $(i,j) \in E^*$ for all $i,j \in U$. Thus $U$ is a $m$-clique of $G(D)$. 

$(\Leftarrow)$ Let $U \subset V$ be a $m$-clique in $D(G)$. Then all pairs $i,j \in U$ are distinguishable in $G$. Thus $U$ is a distinguishable set in $G$ of size $m$. 
\end{proof}

The following Corollary is  immediate.

\begin{cor}
A graph $G$ has a distinguishability $m$ if and only if the size of a maximum clique in $D(G)$ is $m$.
\end{cor}

We now show that finding the distinguishability is 
$\mathcal{NP}$-complete.

\begin{lem}
$\textsc{Distinguishability } \in \mathcal{NP}$
\end{lem}
\begin{proof} 
Let $S \in \textsc{Distinguishability}$ be an instance of distinguishability with graph $G = (V,E,\ell)$, label set $L$, and number $m$. Let the certificate for $S$ be a list of $m$ vertices that constitute a distinguishable set $U$. Clearly this certificate has length polynomial in $|V|$.  A deterministic algorithm to verify this certificate is to check if $i$ and $j$ are distinguishable  for every $i,j \in U$ by iterating over all mutual neighbors of $i$ and $j$. This algorithm has running time polynomial in $|V|$. Therefore, $\text{\textsc{Distinguishability}} \in \mathcal{NP}$.
\end{proof}

\begin{prob}[Clique Number]
Given a simple undirected graph $G$ and a number $m \in \N$, decide if $G$ has a clique of size $m$. Let \textsc{Clique} be the associated decision problem represented as a language.
\end{prob}

\begin{lem}
$\textsc{Distinguishability}$ is $\mathcal{NP}$-hard
\end{lem}
\begin{proof}
We proceed by showing a many-to-one deterministic polynomial-time reduction of \textsc{Clique}, which is $\mathcal{NP}$-complete \cite{Karp1972}, to \textsc{Distinguishability}. That is, we show a map from instances of \textsc{Clique} to instances of \textsc{Distinguishability} that is efficiently computable, and that maps `yes' instances of \textsc{Clique} to `yes' instances of \textsc{Distinguishability} and maps `no' instances of \textsc{Clique} to `no' instances of \textsc{Distinguishability}. 

Consider an  arbitrary instance of \textsc{Clique}  with a simple undirected input graph $G = (V,E,\emptyset)$ and input number $m$. We aim to construct a new graph which has distinguishability equal to the size of the largest clique in $G$. Consider the undirected graph $H = (E\cup V,E_{H}, \ell)$ with label set $L = V$ where 
\begin{equation}
(i,(i,j)), (j,(i,j)) \in E_{H} \text{ iff } (i,j) \in E
\end{equation}
and  
\begin{equation}
    \ell((i,j),i) =  \ell(i,(i,j)) = i.
\end{equation}  In other words for each (undirected) edge $(i,j) \in E$ of the graph $G$, we assign directed edges  from $i\in V$ to $(i,j) \in E$ and
from  $i\in V$ to $(i,j) \in E$ in the graph $H$. There are no edges in $H$ between $i\in V$ and $j \in V$  and no edges between $(i,j) \in E$ and $(k,s) \in E$ and therefore $H$ is bipartite with partition of vertices  $(V,E)$. 

We now show the distinguishability of $H$ is $m$. Consider the distinguishability graph $D(H) = (E\cup V, E_{D(H)}, \emptyset)$. First notice that, because $H$ is bipartite, there is no edge in $D(H)$ between a vertex $j \in V$ and a vertex  $(i,k) \in E$ as $(i,k)$ and $j$ cannot have a mutual neighbor in $H$. Furthermore, there is no edge in $D(H)$ between two vertices $(i,k),(r,s) \in E$ because if $(i,k)$ and $(r,s)$ have a mutual neighbor $j$ then $j$ is not a distinguisher of $(i,k)$ and $(r,s)$ due to all edge labels being identical, i.e.
\begin{align*}
     j &= \ell(j, (r,s)) 
    \\ &= \ell((i,k),j) 
    \\ &= \ell((r,s), j) 
    \\ &= \ell(j, (i,k)).
\end{align*}

Now we show for any $i,j \in V$ there is an edge $(i,j) \in  E_{D(H)}$ if and only if $(i,j) \in E$. 
If $(i,j) \in E$ then $(i,(i,j)) , (j,(i,j)) \in E_H$. Also $ \ell(i,(i,j))  = i$  and $ \ell(j,(i,j)) = j$ so $ \ell(i,(i,j)) \neq \ell(j,(i,j))$ which means $i$ and $j$ are distinguishable in $H$. Therefore, $(i,j) \in E_{D(H)}$. 

Now suppose $(i,j) \notin E$. Then $i$ and $j$ have no mutual neighbors in $H$ and so they are not distinguishable in $H$. This implies $(i,j) \notin E_{D(H)}$. We have shown that the only edges in $D(H)$ connect $i, j \in V$ such that $(i,j) \in E$.  Furthermore, since $(i,j) \in E_{D(H)}$ if and only if $(i,j) \in E$, the subgraph of $D(H)$ induced by $V$ is isomorphic to $G$. Therefore, there is a $m$-clique in $D(H)$ if and only if there is a $m$-clique in $G$.  

The many-to-one reduction is given by the function 
\begin{equation}
    \phi: (G,m) \mapsto (H,m). 
\end{equation}
This function, which amounts to constructing the graph $H$,  can be computed by an algorithm which iterates over the set  $(V \cup E) \times (V \cup E)$ of possible edges in $H$. This algorithm takes time polynomial in $|V| + |E|$, and so $\phi$ is efficiently computable. 
\end{proof}

Note that the graph $H$ is undirected so the completeness holds even if the problem \textsc{Distinguishability} is restricted to undirected graphs. 

\begin{cor}
$\textsc{Distinguishability}$ is ${\mathcal{NP}\textit{-complete}}$. 
\end{cor}

\section{SI: Numerical Simulations }

\label{app:numerical}

\begin{figure}

\begin{subfigure}{.45\textwidth}
  \centering
  \includegraphics[width=\linewidth]{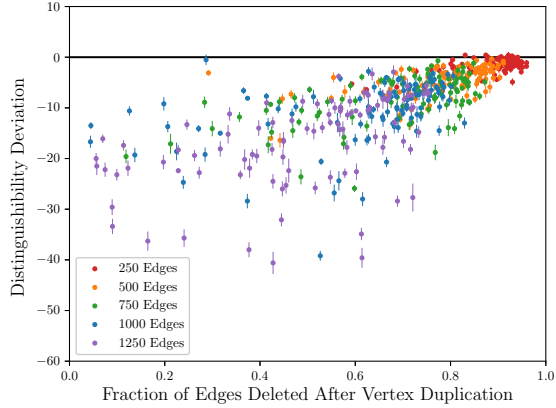}
      \caption{500 graphs, each with 250 vertices, generated by taking evolved graphs $G_i$ (as in Figure \ref{fig:sfig1}) and generating a new graph $J_i$ which has the same signed degree distribution of $G_i$ but is otherwise randomized. Colors and axis are same as Figure \ref{fig:sfig1}.   }
  \label{fig:sfig3}
\end{subfigure}
\begin{subfigure}{.45\textwidth}
  \centering
  \includegraphics[width=\linewidth]{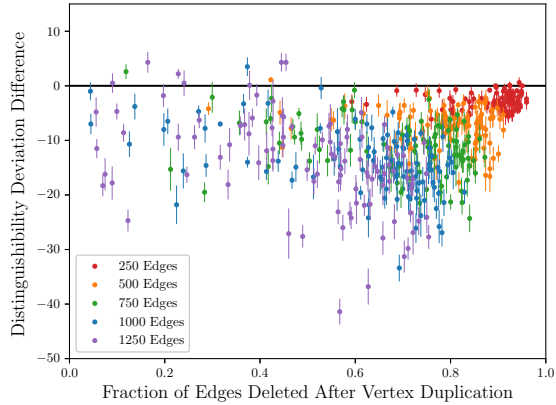}
  \caption{Point-wise difference between distinguishability change in Figure \ref{fig:sfig1} and in Figure \ref{fig:sfig3}. This difference shows that change in distinguishability of the evolved graphs in Figure \ref{fig:sfig1} which cannot be attributed to single vertex characteristics.
}  \label{fig:sfig4}
\end{subfigure}
\caption{Distinguishability deviation of directed graphs. }
\label{fig:fig}
\end{figure}
In this section we give a complete description of our numerical simulations and give evidence that negative distinguishability deviation cannot be explained solely through a graph's signed degree distribution or by small world properties.  

We first describe the procedure for generating the evolved graphs and their corresponding ER-graphs. The distinguishability deviation of these graphs are shown in Figure \ref{fig:sfig1}. Our implementation of this procedure, which can be found at \cite{gitcode}, employs the NetworkX Python package \cite{networkx}. The following description uses the convention $[n] := \{1,2,\dots,n\}$. 
\begin{enumerate}
\setlength{\itemsep}{6pt}
    \item Randomly generate 500 ER-graphs, each with 25 vertices where the fraction of positive edges are chosen uniformly at random from $(.25, .75)$ and the edge density $2|E|/(|V| - 1)|V|$ is chosen uniformly at random from $[1/2(25), 2/25]$ (rounding up to the nearest whole edge).
    
    \item Perform duplication on a random vertex 225 times to each graph generating 500 graphs each with 250 vertices.
    
    \item Divide this set of 500 graphs into 5 sets of 100. Randomly remove edges from graphs in the first set until a final edge count of 250 is reached for each graph.  Repeat for the last four sets using final edge counts of 500, 750, 1000, and 1250 respectively. We refer to the set of graphs  $\{G_i\}_{i \in [500]}$ as the evolved graphs.
    
    \item For each evolved graph $G_i$, calculate its distinguishability $\mathtt D(G_i)$. 
    
    \item \label{stp:expected} For each evolved graph $G_i$ randomly generate 10 new graphs $ G_{i,j}$ with probability $P_i( G_{i,j})$ (the probability distribution in Definition \ref{def:expected dist}) that have the same adjacencies but with a random edge labeling.  Estimate their expected distinguishability by
    \begin{equation}
       \langle \mathtt D(G_i) \rangle \approx \langle \mathtt D(G_i) \rangle_\approx := \frac{1}{10} \sum
       _{j \in [10]}\mathtt D(G_{i,j}).
    \end{equation}
    
    \item Calculate the approximate distinguishability deviation  $\mathtt D(G_i) - \langle \mathtt D(G_i) \rangle_\approx$ of the evolved graphs.  
    
    \item For each evolved graph $G_i$, randomly generate an ER-graph $H_i$ with the same number of vertices, edges, and positive and negative labels as $G_i$, and calculate its distinguishability $\mathtt D(H_i)$.
    
    \item For each graph $H_i$ compute $\langle H_i \rangle_\approx$ as in Step \ref{stp:expected}. 
    
    \item  Calculate the distinguishability deviation  $\mathtt D(H_i) - \langle \mathtt D(H_i) \rangle_\approx$ of the ER-graphs and compute the standard deviation.  
\end{enumerate}
In Figures \ref{fig:sfig1} and \ref{fig:und fig1}, point color indicates final number of edges after edge deletion of the evolved graphs $\{G_i\}$. Grey points represent the ER-graphs $\{H_i\}$. The vertical axis denotes distinguishability deviation. The horizontal axis gives the fraction of edges which are removed in the deletion process. Note that the above procedure applies to both directed and undirected graphs, the data for which are given in Figures \ref{fig:sfig1} and \ref{fig:und fig1} respectively.

It is natural to ask if these distinguishability deviations can be explained entirely through a graph's signed degree distribution. The idea is that distinguishers must have two edges of differing sign which, in the directed case, are either both incoming or both outgoing. As a result, we expect graphs for which most vertices have all edges of the same sign to have low distinguishability. When the edge labels of these networks are randomized, the homogeneity of signed degree can be removed, potentially creating higher distinguishability.

To address this question we compute, for each evolved graph, a random network which has the same signed degree as the original network. To generate these networks we use the following algorithm in which we randomly connect edge stubs of matching sign. Starting from a list of edge stubs for each vertex and a graph with no edges, we randomly choose two edge stubs, remove them from their respective lists, and add an edge between their respective vertices. Note that when we randomly choose edge stubs we do not allow choice of edge stubs where the introduced  edge  would create a multigraph. In the undirected case this means we do not pick edge stubs between already connected vertices.

If the only edge stubs that remain are on two vertices $v_1$ and $u_1$ such that adding a edge between these vertices would create a muligraph, a random rewiring is performed. That is, a randomly chosen previously added edge $(v_2,u_2)$ is removed from the graph and the edges $(v_1,u_2)$ and $(v_2,u_1)$ are added. If for all such edges $(v_2,u_2)$ this rewiring would create a multigraph, then the random graph generation is restarted. This process is used for both directed and undirected graphs with the difference being that for directed graphs \textit{in}-edge stubs are matched with \textit{out}-edge stubs. Python code for generating these graphs is provided at \cite{gitcode}.

For directed graphs, we plot distinguishability deviation of the graphs generated by this procedure in Figure \ref{fig:sfig3}. Note that each point in this figure is in one-to-one correspondence with the colored points of Figure \ref{fig:sfig1}. From this data we see that the signed degree preserved randomizations exhibit significant negative distinguishability deviation. However, this negative distinguishability deviation is less strong than the distinguishability deviation observed for the evolved graphs. Investigating further, Figure \ref{fig:sfig4} shows the point-by-point difference between distinguishability deviations of the evolved graphs and their randomized signed degree-preserving counterparts. From this figure it is apparent that, almost always, the randomized versions of the evolved graphs exhibit lower distinguishability deviation. These results are replicated in Figures \ref{fig:und fig3} and \ref{fig:und fig4} for undirected graphs. We conclude that the large distinguishability deviation observed in the evolved graphs can not be explained solely through signed degree distribution. 

\begin{figure}
\begin{subfigure}{.45\textwidth}
 \centering
 \includegraphics[width=\linewidth]{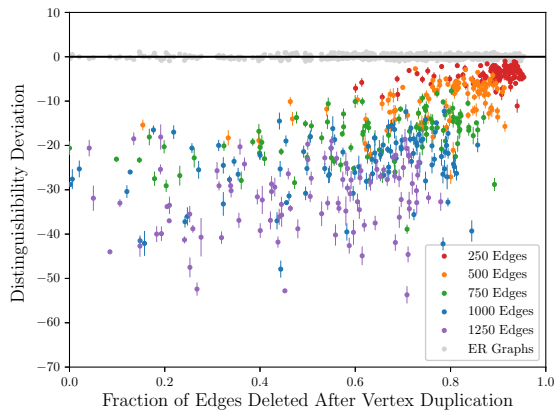}
 \caption{Same as Figure \ref{fig:sfig1} but with undirected graphs}
 \label{fig:und fig1}
\end{subfigure}
\begin{subfigure}{.45\textwidth}
  \centering
  \includegraphics[width=\linewidth]{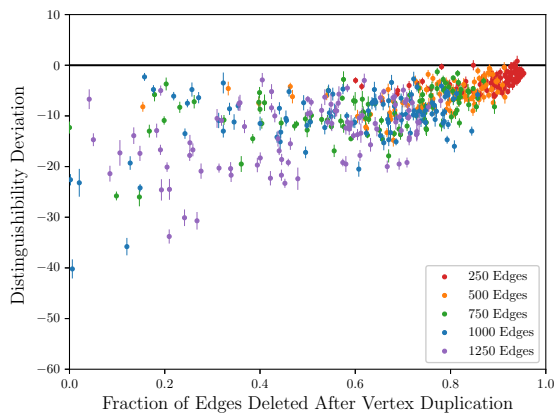}
  \caption{Same as Figure \ref{fig:sfig3} but with undirected graphs}
  \label{fig:und fig3}
\end{subfigure}
\begin{subfigure}{.45\textwidth}
  \centering
  \includegraphics[width=\linewidth]{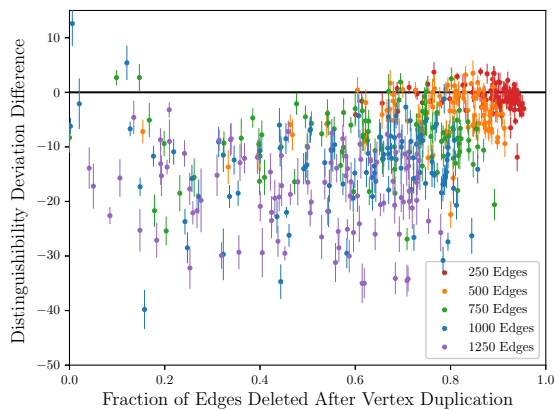}
  \caption{Point-wise difference between distinguishabilities in Figure \ref{fig:und fig1} and in Figure \ref{fig:und fig3}
}  \label{fig:und fig4}
\end{subfigure}
\caption{Distinguishability deviation of undirected graphs. }
\label{fig:und fig}
\end{figure}

We also computed the distinguishability deviation in the experimentally derived networks of \cite{vin14} and \cite{Collombet2017} along with their signed degree preserved randomizations, as shown in Table \ref{tab:results}. These results agree with simulations since both networks exhibited stronger negative distinguishability deviation than their preserved sign degree sequence randomizations. We conclude that the signed degree distribution of a network can not entirely predict its distinguishability deviation.

This table also includes the distinguishability deviation of both directed and undirected Erd{\"o}s-R{\'e}nyi graphs (ER) and Watts-Strogatz graphs \cite{Watts1998} with characteristics similar to the published biological networks.
For the ER-graphs, number of vertices and number of positive and negative edges are the same as the corresponding biological networks. For the Watts-Strogatz graphs,  we picked the number of vertices and mean degree $k$ to be the same as the biological networks.  We used a rewiring probability of $\beta = .1$ to target the small-world regime of low path length and high clustering observed in \cite{Watts1998}. In the Watts-Strogatz model we randomly assigned edge signs with the probability at which to occur in the original network.
 For both models, direction is randomly assigned to edges when generating directed graphs. 

Since the generation of these  random  graph models assigns edge labels randomly, we expect near zero average distinguishability deviation. However, we are interested in the standard deviation of the distinguishability deviation since this describes the likelihood to produce outliers with large negative distinguishability deviation in these random models. 
The observed small standard deviation suggests that these models are unlikely to produce graphs with distinguishability deviation near that observed in the biological networks. The distinguishability deviation of graphs generated by the Watts-Strogatz model is nearly the same as the ER-graphs, suggesting that small world properties have little to no effect on distinguishability deviation.

\begin{table*}[t]
\centering
\begin{tabular}{l|l|l|l|l|}
\cline{2-5}
                                              & \multicolumn{2}{c|}{D. Melanogaster} & \multicolumn{2}{c|}{Blood Cell} \\ \cline{2-5} 
                                              & Dist.             & Dist. Deviation   & Dist.          &  Dist. Deviation \\ \hline
\multicolumn{1}{|l|}{Original graphs}         & 7                 & $-24.2 \pm .7$    & 4              & $-1.6 \pm 0.6$  \\ \hline
\multicolumn{1}{|l|}{Preserved signed degree} & $4.94 \pm 0.4$    & $-1.0 \pm 0.6$     & $4.9 \pm .6$   & $-0.3 \pm .8$    \\ \hline
\multicolumn{1}{|l|}{ER-graph}                & $3.0 \pm 0.2$     & $0.0 \pm .2$     & $3.4 \pm .5$  & $-0.1 \pm 0.7$   \\ \hline
\multicolumn{1}{|l|}{Watts-Strogatz}          & $5.0 \pm 0.1$     & $0.0 \pm .2$     & $3.0 \pm 0.4$  & $0.0 \pm .6$  \\ \hline
\end{tabular}
\caption{Distinguishability and distinguishability deviation for two experimentally derived networks and random graph models. For random graph models, values are an average over 100 random graphs. Note, the blood cell network contained a single multi-edge which was ignored in the calculation of these values. Graphs described in the first D. Melanogaster column are undirected and graphs described in the Blood Cell column are directed.} \label{tab:results}
\end{table*}
\FloatBarrier

\end{document}